\documentclass[aps,prl,twocolumn,superscriptaddress,showpacs,10pt]{revtex4}

\usepackage[T1]{fontenc}
\usepackage{times}
\usepackage{color,graphicx}
\usepackage{array}
\usepackage{amsthm,amssymb,amsmath}
\usepackage{tikz} % Required for drawing custom shapes
\usepackage{breakurl}
\usetikzlibrary{arrows.meta}
\usetikzlibrary{shapes}
\usetikzlibrary{calc,fadings,decorations.pathreplacing,angles,positioning,fit,backgrounds,fpu}

\definecolor{myblue}{RGB}{0,0,128}
\definecolor{myblue2}{RGB}{0,32,96}
\definecolor{myblue3}{RGB}{0,64,64}
\definecolor{myblue4}{RGB}{0,96,32}
\definecolor{myblue5}{RGB}{0,128,0}
\definecolor{myblue6}{RGB}{32,96,0}
\definecolor{myblue7}{RGB}{64,64,0}
\definecolor{myblue8}{RGB}{96,32,0}
\definecolor{myblue9}{RGB}{128,0,0}
\definecolor{myblue10}{RGB}{96,0,32}
\definecolor{myblue11}{RGB}{64,0,64}

\newcommand\ket[1]{\ensuremath{|#1\rangle}}
\newcommand\bra[1]{\ensuremath{\langle#1|}}
\newcommand{\braket}[2]{\langle #1 | #2 \rangle}
\newcommand{\ketbra}[2]{| #1 \rangle\langle #2 |}

\newcommand\tr{\mathop{\rm tr}\nolimits}

\def\C{\mathbb{C}}

\newcommand{\bes}{\begin{eqnarray*}}
	\newcommand{\ees}{\end{eqnarray*}} 
\newcommand{\bpm}{\begin{pmatrix}}
	\newcommand{\epm}{\end{pmatrix}}

\def\cI{{\mathcal I}}

\def\diag{{\rm diag}\,}
\def\tr{{\rm tr}\,}

\newcommand{\defeq}{\stackrel{\smash{\textnormal{\tiny def}}}{=}}

\newtheorem{theorem}{Theorem}

\begin{document}

%% End-Of-Header

\title{The modified trace distance of coherence is constant on most pure states}

\author{Nathaniel Johnston}
\affiliation{Department of Mathematics and Computer Science, Mount Allison University, Sackville, NB, Canada E4L 1E4}
\affiliation{Department of Mathematics and Statistics, University of Guelph, Guelph, ON, Canada N1G 2W1}
\author{Chi-Kwong Li}
\affiliation{Department of Mathematics, College of William and Mary,  Williamsburg, VA, USA  23187}
\author{Sarah Plosker}
\affiliation{Department of Mathematics \& Computer Science, Brandon University, Brandon,
	MB, Canada R7A 6A9}
	\affiliation{Department of Mathematics and Statistics, University of Guelph, Guelph, ON, Canada N1G 2W1}

\begin{abstract}
    Recently, the much-used trace distance of coherence was shown to not be a proper measure of coherence, so a modification of it was proposed. We derive an explicit formula for this modified trace distance of coherence on pure states. Our formula shows that, despite satisfying the axioms of proper coherence measures, it is likely not a good measure to use, since it is maximal (equal to $1$) on all except for an exponentially-small (in the dimension of the space) fraction of pure states.
\end{abstract}

\date{\today}


\maketitle

\section{I. Introduction}

There are two key resources that set quantum information theory apart from classical information theory: coherence and entanglement. A necessary requirement for entanglement is that the quantum system can be viewed as multiple quantum systems that interact. For over two decades, entanglement overshadowed coherence in the literature  \cite{BDSW96,Bra05,HHT01,Hor01,Rai99,Shi95,Ste03,VT99}. Now, much interest has percolated with respect to coherence, especially as it pertains to quantum optics \cite{Glau63, Su63}, quantum biology \cite{bio1, bio2, bio3}, and thermodynamics \cite{thermo1, thermo2, thermo3, thermo4}, as well as other areas. Multiple quantum systems are not required for coherence, and it arises in any system that cannot be reduced to a classical one (essentially, whenever one deals with superpositions of quantum states). 

Formalizing the task of measuring coherence began in \cite{Abe06}, which effectively gave a one-to-one correspondence between coherence measures and entanglement measures. Later, a framework \cite{BCP14}  of four defining properties for a coherence measure to be \emph{proper} was introduced, with these four properties being seen as highly desirable, or even required, for valid coherence measures. Three of the most commonly used  measures of coherence, namely the $\ell_1$-norm of coherence, the relative entropy of coherence, and the robustness of coherence \cite{NBCPJA16}, have all been shown to be proper coherence measures. Another measure coherence, called the trace distance of coherence \cite{RPL15}, has also received quite a bit of attention, but was recently shown to not be a proper coherence measure in general \cite{YZXT16} (although it is proper when restricted to qubit states or $X$~states). This led to a ``modified'' trace distance of coherence being proposed in \cite{YZXT16}, which was shown to indeed be a proper coherence measure, and has been further studied in \cite{CF17,HHWPZF17}.

Despite the modified trace distance of coherence being a proper measure of coherence, we demonstrate that it is not very useful, since it is equal to its maximal value (i.e., $1$) for all but an exponentially-small proportion of pure states. We also provide several related results along the way, such as an explicit formula for the modified trace distance of coherence on pure states, and we show that the closest incoherent state to a pure state in this measure can always be chosen to have just one non-zero entry. We also demonstrate numerically that similar results likely hold for density matrices with a fixed rank larger than $1$.

In Section~II, we review preliminary definitions and notation needed for the remainder of the paper. In Section~III, we present our main results on the modified version of the trace distance of coherence: we show  the non-uniqueness of  the incoherent states $\tilde \delta$ closest to a given state $\rho$, we give a formula for computing the modified trace distance of coherence for pure states and describe an optimal $\tilde p$ and $\tilde \delta$, and we show that the modified trace distance of coherence is equal to $1$---its maximum possible value---on all except for an exponentially-small (in the dimension of the space) fraction of pure states. We numerically extend our results to mixed states of higher rank in Section~IV, and we provide concluding remarks in Section~V.

\section{II. Preliminaries and the trace distance of coherence}\label{sec:prelim}

Let $\cI$ be the set of diagonal density matrices (incoherent states). For any density matrix $\rho$, the trace distance of coherence is defined as the trace norm distance between $\rho$ and the closest incoherent state: 
 \begin{eqnarray*}
C_{\tr}(\rho) \defeq  \min_{\delta\in \mathcal{I}}\|\rho-\delta\|_{\tr}=\min_{\delta\in \mathcal{I}}\sum_{i=1}^n|\lambda_i(\rho-\delta)|,
\end{eqnarray*}
where $\lambda_i(\rho-\delta)$ are the eigenvalues of the matrix $\rho-\delta$.

For a coherence measure $C$ to be a proper measure of coherence, it must satisfy the following four conditions \cite{BCP14}:
\begin{enumerate}
    \item[(1)] $C(\rho) \geq 0$, with equality if and only if $\rho\in \cI$;
    \item[(2)] $C(\rho)\geq C(\Lambda(\rho))$ if $\Lambda$ is an incoherent operation, i.e.\ a completely positive trace preserving linear (CPTP) map $\Lambda(\rho)=\sum_iK_i\rho K_i^\dagger$ whose Kraus operators satisfy $K_i\mathcal{I}K_i^\dagger \subset\mathcal{I}$; 
    \item[(3)] $C(\rho)\geq \sum_jp_jC(\rho_j)$ where $p_j=\tr(K_j\rho K_j^\dagger)$, $\rho_j=(K_j\rho K_j^\dagger)/p_j$, and $\{K_j\}$ is a set of incoherent Kraus operators (that is, Kraus operators that satisfy $K_i\mathcal{I}K_i^\dagger \subset\mathcal{I}$); and
    \item[(4)] $\sum_jp_jC(\rho_j)\geq C(\sum_jp_j\rho_j)$ for any set of  states $\{\rho_j\}$ and any probability distribution $\{p_j\}$.
\end{enumerate}

Items (3) and (4) above (monotonicity   under  selective  measurements  on  average and non-increasing under mixing of quantum states (convexity), respectively) are equivalent to $C(p_1\rho_1\oplus p_2\rho_2)=p_1C(\rho_1)+p_2C(\rho_2)$ for all block-diagonal states $\rho$ in the incoherent basis \cite{YZXT16}; this equation is more readily manipulated and was shown to be violated for the trace distance of coherence (and in particular, condition~(3) above does not hold for the trace distance of coherence).

In order to address this problem, the following \emph{modified trace distance of coherence} (originally called ``modified trace norm of coherence'') was proposed \cite{YZXT16}:
\begin{equation}\label{defn:mod_tr_dist}
    C^\prime_\tr(\rho) \defeq \min \big\{\|\rho- p\delta\|_\tr: \delta \in \cI, p \in [0, \infty) \big\}.
\end{equation}
The advantage of this measure is that it really is a proper coherence monotone (i.e., it satisfies conditions (1)--(4) above), while still retaining the ``spirit'' of the trace distance of coherence.

\section{III. Main results}\label{sec:modified}

In the $n = 2$ case, the modified trace distance of coherence coincides with the familiar \emph{$\ell_1$-norm of coherence}:
\begin{eqnarray*}
C_{\ell_1}(\rho) & \defeq & \sum_{i\neq j} |\rho_{ij}|.
\end{eqnarray*}

While this fact is already known \cite{CF17}, we state it as a theorem and include an alternate proof below that illustrates the non-uniqueness of ``the'' state $\tilde \delta\in \cI$ attaining the minimum of $C^\prime_\tr$ for a given $\rho$. This non-uniqueness is important, as it plays a role in the pure state result that we will prove shortly.

\begin{theorem}\label{thm:mod_tr_coh_qubit}
    If $\rho \in M_2(\mathbb{C})$ then $C^\prime_\tr(\rho) = C_{\ell_1}(\rho) = 2|\rho_{12}|$.
\end{theorem}
\begin{proof}
    We compute
    \begin{eqnarray*}
    && \|\rho - p\cdot\diag(d_1, d_2)\|_\tr 
    = \left\|\begin{pmatrix}
    \rho_{11} - p d_1 & \rho_{12} \cr \rho_{21} & \rho_{22} - pd_2\cr
    \end{pmatrix}\right\|_\tr \\
    &=& \left\|\begin{pmatrix} \rho_{12} & \rho_{22} - pd_2\cr
    \rho_{11} - p d_1 & \rho_{21} \cr 
    \end{pmatrix}\right\|_\tr \ge 2|\rho_{12}|,
    \end{eqnarray*}
    where the inequality holds with equality if and only if 
    $\rho_{22} - pd_2 =  \rho_{11} - p d_1 = \mu$ for some $\mu$ with $|\mu| \le |\rho_{12}|$.
    So, the optimal solution set for $\diag(dp_1, dp_2)\in \cI$ equals
    $$\big\{\diag(\rho_{11} - \mu, \rho_{22} - \mu): -|\rho_{12}| \le \mu \le \min\{\rho_{11},\rho_{22}\}\big\}.$$
    In particular, we can choose $\mu = 0$ to get $\delta = \diag(\rho_{11},\rho_{22})$ so that 
    $$\rho - \diag(pd_1, pd_2) = \rho_{12}E_{12}+\rho_{21}E_{21},$$ where $E_{ij}$ is an appropriately sized matrix with 1 in the $(i,j)$-th position and zeros elsewhere. The above expression has trace norm equal to $2|\rho_{12}|$, which equals $C_{\ell_1}(\rho)$.
\end{proof}

The above proof shows that the incoherent state attaining the minimum in the modified trace distance of coherence is very non-unique in general. For example, instead of choosing $\mu = 0$ like we did in the proof, we could have chosen $\mu = -|\rho_{12}|$ to get the incoherent state $\diag(\rho_{11} + |\rho_{12}|, \rho_{22} + |\rho_{12}|)$ so that
\[
    \rho - \diag(pd_1, pd_2) = -|\rho_{12}|I + \rho_{12}E_{12} + \rho_{21}E_{21},
\]
which also has trace norm equal to $2|\rho_{12}|$.

Next, we present an explicit method of computing $C^\prime_\tr(\ketbra{x}{x})$ (i.e., the modified trace distance of coherence when restricted to pure states) that is analogous to the method that was derived in \cite{cjlp} for the (standard) trace distance of coherence. It turns out that the formulas involved for this version of the trace distance are actually significantly simpler than they were for the original version. In particular, we show that $C^\prime_\tr(\ketbra{x}{x})$ is simply a function of the largest entry of $\ket{x}$, which we denote by $\|\ket{x}\|_{\infty} := \max_{j} \{|x_j|\}$.

\begin{theorem}\label{thm:C'mu}
    Suppose $\ket{x} \in \mathbb{C}^n$ is a pure state.
    \begin{enumerate}
        \item[a)] If $\|\ket{x}\|_{\infty} \leq 1/\sqrt{2}$ then $C^\prime_\tr(\ketbra{x}{x}) = 1$. Furthermore, an optimal $p$ in~\eqref{defn:mod_tr_dist} is $\tilde{p} = 0$.
        
        \item[b)] If $\|\ket{x}\|_{\infty} > 1/\sqrt{2}$ then \[{} \quad \quad C^\prime_\tr(\ketbra{x}{x}) = 2\|\ket{x}\|_{\infty}\sqrt{1 - \|\ket{x}\|^2_{\infty}}.\]Furthermore, an optimal $p$ and $\delta$ in~\eqref{defn:mod_tr_dist} are
        \[
            {} \quad \quad \tilde{p} = 2\|\ket{x}\|^2_{\infty} - 1 \quad \text{and} \quad \tilde{\delta} = \mathrm{diag}(1,0,0,\ldots,0),
        \]
        respectively.
    \end{enumerate}
\end{theorem}

Before proving Theorem~\ref{thm:C'mu}, we note that $C^\prime_\tr(\rho)$ can be computed numerically via the following semidefinite program for arbitrary mixed states $\rho$:
\begin{equation}\begin{aligned}\label{sdp:mod_coh_primal}
	\text{minimize:} & \ \ \|\rho - D\|_\tr \\
	\text{subject to:} & \ \ D \ \text{diagonal} \\
	& \ \ D \succeq O.
\end{aligned}\end{equation}
However, Theorem~\ref{thm:C'mu} provides a much more explicit way of dealing with this quantity when $\rho = \ketbra{x}{x}$ is pure. We also note that we only describe ``an'' optimal $p$ and $\delta$ in the statement of the theorem (as opposed to ``the'' optimal $p$ and $\delta$), since the points attaining the minimum may not be unique, as we noted after Theorem~\ref{thm:mod_tr_coh_qubit}.

\begin{proof}
    We prove the result by showing that in each of case~(a) and~(b), $C^\prime_\tr(\ketbra{x}{x})$ is bounded both above and below by the indicated quantity. Also, we may assume without loss of generality that $\ket{x} = (x_1, \dots, x_n)^t \in \mathbb{R}^n$ with $x_1 \ge \cdots \ge x_n\ge 0$. This follows from the fact that the modified trace distance of coherence is invariant under maps of the form $\ket{x} \mapsto PU\ket{x}$, where $P$ is a permutation matrix and $U$ is a diagonal unitary matrix. In particular, this means that $\|\ket{x}\|_{\infty} = x_1$.
    
    In case~(a) of the theorem, $C^\prime_\tr(\ketbra{x}{x}) \leq 1$ trivially since we can choose $p = 0$ in definition~\eqref{defn:mod_tr_dist}. Similarly, in case~(b) we can choose
    \[
        p = 2x_1^2 - 1 \quad \text{and} \quad \delta = \mathrm{diag}(1,0,0,\ldots,0),
    \]
    as suggested by the theorem, and then we have $C^\prime_\tr(\ketbra{x}{x}) \leq \|\rho - p\delta\|_\tr$. To compute $\|\rho - p\delta\|_\tr$ we note that $\mathrm{rank}(\rho - p\delta) \leq 2$, and it is straightforward to verify that its non-zero eigenvalues are \[ \lambda_{\pm} = (1-x_1^2) \pm x_1\sqrt{1-x_1^2}, \]with corresponding (unnormalized) eigenvectors
    \begin{equation}\label{eq:vpm}
        \vec{v}_{\pm} = \big(\pm \sqrt{1-x_1^2}, x_2, x_3, \ldots, x_n\big)^t,
    \end{equation}respectively. Thus \[\|\rho - p\delta\|_\tr = |\lambda_+| + |\lambda_{-}| = 2x_1\sqrt{1-x_1^2},\]which establishes the desired $C^\prime_\tr(\ketbra{x}{x}) \leq 2x_1\sqrt{1-x_1^2}$ inequality.
    
    To obtain the opposite inequalities, we use weak duality applied to the semidefinite program~\eqref{sdp:mod_coh_primal}. In particular, a routine calculation shows that the dual of that semidefinite program has the form
    \begin{equation}\begin{aligned}\label{sdp:mod_coh_dual}
    	\text{maximize:} & \ \ -\bra{x}(Y + Y^*)\ket{x} \\
    	\text{subject to:} & \ \ \begin{bmatrix}X & Y \\ Y^* & Z\end{bmatrix} \succeq O \\
    	& \ \ \|X\|,\|Z\| \leq 1/2 \\
    	& \ \ \mathrm{diag}(Y) = \vec{0}.
    \end{aligned}\end{equation}
    Any feasible point that we can find for this dual problem immediately (by weak duality) gives a lower bound on $C^\prime_\tr(\ketbra{x}{x})$.
    
    In case~(a) of the theorem, we can choose $X = Z = I/2$. To see how we choose $Y$, first note that there exists a particular pure state $\ket{y}$ with the property that $|y_j| = |x_j|$ for all $j$ and $\braket{x}{y} = 0$. To construct such a $\ket{y}$, we just need to choose the phases $e^{i\theta_j}$ of each entry $y_j$ of $\ket{y}$, and we want them to satisfy
    \[
        \braket{x}{y} = \sum_{j=1}^n e^{i\theta_j}|x_j|^2 = 0.
    \]
    Well, since $|x_j|^2 \leq 1/2$ for all $j$, such phases do indeed exist (this is basically just the triangle inequality---we can choose $e^{i\theta_1} = 1$ and then choose the other phases so as to work our way back to the origin in the complex plane).
    
    Now that we have $\ket{y}$, we choose $Y = (\ketbra{y}{y} - \ketbra{x}{x})/2$. Then (since $|y_j| = |x_j|$ for all $j$) we have $\mathrm{diag}(Y) = \vec{0}$. Furthermore, $\ket{x}$ and $\ket{y}$ are orthogonal, so the eigenvalues of $Y$ are $\pm 1/2$ and some zeroes, so $\begin{bmatrix}I/2 & Y \\ Y^* & I/2\end{bmatrix} \succeq O$. Thus all of the constraints of the SDP~\eqref{sdp:mod_coh_dual} are satisfied, and the corresponding objective value is
    \[-\bra{x}(Y + Y^*)\ket{x} = |\braket{x}{x}|^2 - |\braket{x}{y}|^2 = 1 - 0 = 1,\]
    which establishes the desired lower bound $C^\prime_\tr(\ketbra{x}{x}) \ge 1$, and completes the proof of part~(a) of the theorem.
    
    On the other hand, in case~(b) of the theorem, we can choose $X = Z = I/2$ and $Y = (\ketbra{v_{-}}{v_{-}} - \ketbra{v_{+}}{v_{+}})/2$, where $\ket{v_{\pm}}$ is the normalization of the vectors $\vec{v}_{\pm}$ from Equation~\eqref{eq:vpm}:
    \[
        \ket{v_{\pm}} := \frac{1}{\sqrt{2}\sqrt{1 - x_1^2}}\vec{v}_{\pm}.
    \]
    By construction, $Y$ has eigenvalues $\pm 1/2$ and some zeroes, so $\begin{bmatrix}I/2 & Y \\ Y^* & I/2\end{bmatrix}$ is indeed positive semidefinite. The only other constraint to be checked is that $\mathrm{diag}(Y) = \vec{0}$, and this follows from the fact that the entries of $\ket{v_{-}}$ and $\ket{v_{+}}$ have the same absolute values as each other.
    
    Thus $X,Y,$ and $Z$ define a feasible point of the SDP~\eqref{sdp:mod_coh_dual}, and the corresponding objective value is easily computed to be\[-\bra{x}(Y + Y^*)\ket{x} = |\braket{x}{v_{+}}|^2 - |\braket{x}{v_{-}}|^2 = 2x_1\sqrt{1-x_1^2},\]which establishes the desired lower bound $C^\prime_\tr(\ketbra{x}{x}) \geq 2x_1\sqrt{1-x_1^2}$ and completes the proof.
\end{proof}

In addition to providing an explicit formula for $C^\prime_\tr(\ketbra{x}{x})$, Theorem~\ref{thm:C'mu} demonstrates that the modified trace distance of coherence may have some limitations as a measure of coherence (despite satisfying the requirements for it to be physically relevant presented in \cite{BCP14}), since it only depends on a the largest entry of $\ket{x}$ (see Figure~\ref{fig:ctr_mod_const}).
\begin{figure}[!htb]
	\centering
	\def\x{\noexpand\x}    % Prevent \x from being expanded inside an \edef
	\edef\ctr{2*max(\x,0.70710678118)*pow(max(0,1 - pow(max(\x,0.70710678118),2)),0.5)}
	\begin{tikzpicture}[scale=5]
		% GRID
		\foreach \x in {1,...,4} \draw[color=gray!25] (0.25*\x,0) -- (0.25*\x,1.1);
		\foreach \y in {1,...,4} \draw[color=gray!25] (0,0.25*\y) -- (1.1,0.25*\y);
		
		% TICKS on AXES
		\draw (0.25,0.5pt) -- (0.25,-1pt) node[anchor=north] {\footnotesize 1/4};
		\draw (0.5,0.5pt) -- (0.5,-1pt) node[anchor=north] {\footnotesize 1/2};
		\draw (0.75,0.5pt) -- (0.75,-1pt) node[anchor=north] {\footnotesize 3/4};
		\draw (1,0.5pt) -- (1,-1pt) node[anchor=north] {\footnotesize 1};
		\draw (0.5pt,0.25) -- (-1pt,0.25) node[anchor=east] {\footnotesize 1/4};
		\draw (0.5pt,0.5) -- (-1pt,0.5) node[anchor=east] {\footnotesize 1/2};
		\draw (0.5pt,0.75) -- (-1pt,0.75) node[anchor=east] {\footnotesize 3/4};
		\draw (0.5pt,1) -- (-1pt,1) node[anchor=east] {\footnotesize 1};
		
		% GRAPH ITSELF
        \draw[color=myblue, domain=0:1.001, samples=201, /pgf/fpu, /pgf/fpu/output format=fixed] 
		plot (\x, {\ctr});
		
		% AXES
		\draw[thick,-to] (-0.1,0) -- (1.1,0) node[anchor=west]{$\|\ket{x}\|_{\infty}$};
		\draw[thick,-to] (0,-0.1) -- (0,1.1) node[anchor=south]{$C^\prime_\tr(\ketbra{x}{x})$};
		\end{tikzpicture}
	\caption{$C^\prime_\tr(\ketbra{x}{x})$ as a function of $\|\ket{x}\|_{\infty}$.}\label{fig:ctr_mod_const}
\end{figure}

Furthermore, $C^\prime_\tr(\ketbra{x}{x})$ is constant and equal to its maximal value for a very large proportion of the state space. For example, if $\ket{x} = (1,1,0,0,\ldots,0)/\sqrt{2}$ and $\ket{y} = (1,1,1,\ldots,1)/\sqrt{n}$ then $C^\prime_\tr(\ketbra{x}{x}) = C^\prime_\tr(\ketbra{y}{y})$, which seems very undesirable. Contrast this with the case of the trace distance of coherence, the robustness of coherence, the relative entropy of coherence, the the $\ell_1$-norm of coherence, all of which attain their maximum values only at the pure states $\ket{x} \in \mathbb{C}^n$ with $|x_j| = 1/\sqrt{n}$ for all $j$.

This problem does not present itself in dimension $n = 2$, since every pure state $\ket{x} \in \C^2$ has $\|\ket{x}\|_{\infty} \geq 1/\sqrt{2}$. However, in higher dimensions, $C^\prime_\tr(\ketbra{x}{x})$ provides no information whatsoever on the vast majority of pure states, since concentration of measure (see \cite{Hay10}, for example) says that the proportion of pure states with $\|\ket{x}\|_{\infty} \geq 1/\sqrt{2}$ decreases exponentially in the dimension, and these are the only pure states with $C^\prime_\tr(\ketbra{x}{x}) \neq 1$. The following theorem quantifies this observation explicitly.

\begin{theorem}\label{thm:proportion_of_useless}
    The proportion (with respect to uniform Haar measure) of pure states $\ket{x} \in \mathbb{C}^n$ for which $C^\prime_\tr(\ketbra{x}{x}) = 1$ is exactly $1-n/2^{n-1}$.
\end{theorem}

For example, this theorem says that already in dimension $n = 19$, over 99.99\% of pure states (chosen according to uniform Haar measure) have $C^\prime_\tr(\ketbra{x}{x}) = 1$ (see Figure~\ref{fig:ctr_mod_dims}).

\begin{figure}[!htb]
	\centering
	\def\x{\noexpand\x}    % Prevent \x from being expanded inside an \edef
	\edef\ctr{2*\x*pow(1 - pow(\x,2),0.5)}
	\begin{tikzpicture}[scale=5]
		% GRID
		\foreach \x in {1,...,4} \draw[color=gray!25] (0.25*\x,0) -- (0.25*\x,1.1);
		\foreach \y in {1,...,4} \draw[color=gray!25] (0,0.25*\y) -- (1.1,0.25*\y);
		
		% TICKS on AXES
		\draw (0.25,0.5pt) -- (0.25,-1pt) node[anchor=north] {\footnotesize 5};
		\draw (0.5,0.5pt) -- (0.5,-1pt) node[anchor=north] {\footnotesize 10};
		\draw (0.75,0.5pt) -- (0.75,-1pt) node[anchor=north] {\footnotesize 15};
		\draw (1,0.5pt) -- (1,-1pt) node[anchor=north] {\footnotesize 20};
		\draw (0.5pt,0.25) -- (-1pt,0.25) node[anchor=east] {\footnotesize 0.9};
		\draw (0.5pt,0.5) -- (-1pt,0.5) node[anchor=east] {\footnotesize 0.99};
		\draw (0.5pt,0.75) -- (-1pt,0.75) node[anchor=east] {\footnotesize 0.999};
		\draw (0.5pt,1) -- (-1pt,1) node[anchor=east] {\footnotesize 0.9999};
		
		% PLOT POINTS
		\draw[color=myblue,fill=myblue] (0.1,0) circle (0.0125); % n = 2
		\draw[color=myblue,fill=myblue] (0.15,0.0312) circle (0.0125); % n = 3
		\draw[color=myblue,fill=myblue] (0.2,0.0753) circle (0.0125); % n = 4
		\draw[color=myblue,fill=myblue] (0.25,0.1263) circle (0.0125); % n = 5
		\draw[color=myblue,fill=myblue] (0.3,0.1817) circle (0.0125); % n = 6
		\draw[color=myblue,fill=myblue] (0.35,0.2403) circle (0.0125); % n = 7
		\draw[color=myblue,fill=myblue] (0.4,0.3010) circle (0.0125); % n = 8
		\draw[color=myblue,fill=myblue] (0.45,0.3635) circle (0.0125); % n = 9
		\draw[color=myblue,fill=myblue] (0.5,0.4273) circle (0.0125); % n = 10
		\draw[color=myblue,fill=myblue] (0.55,0.4922) circle (0.0125); % n = 11
		\draw[color=myblue,fill=myblue] (0.6,0.5580) circle (0.0125); % n = 12
		\draw[color=myblue,fill=myblue] (0.65,0.6246) circle (0.0125); % n = 13
		\draw[color=myblue,fill=myblue] (0.7,0.6918) circle (0.0125); % n = 14
		\draw[color=myblue,fill=myblue] (0.75,0.7596) circle (0.0125); % n = 15
		\draw[color=myblue,fill=myblue] (0.8,0.8278) circle (0.0125); % n = 16
		\draw[color=myblue,fill=myblue] (0.85,0.8965) circle (0.0125); % n = 17
		\draw[color=myblue,fill=myblue] (0.9,0.9656) circle (0.0125); % n = 18
		\draw[color=myblue,fill=myblue] (0.95,1.0349) circle (0.0125); % n = 19
		\draw[color=myblue,fill=myblue] (1,1.1046) circle (0.0125); % n = 20

		% LINE CONNECTING POINTS
		\draw[color=myblue] (0.1,0) -- (0.15,0.0312) -- (0.2,0.0753) -- (0.25,0.1263) -- (0.3,0.1817) -- (0.35,0.2403) -- (0.4,0.3010) -- (0.45,0.3635) -- (0.5,0.4273) -- (0.55,0.4922) -- (0.6,0.5580) -- (0.65,0.6246) -- (0.7,0.6918) -- (0.75,0.7596) -- (0.8,0.8278) -- (0.85,0.8965) -- (0.9,0.9656) -- (0.95,1.0349) -- (1,1.1046);
		
		% AXES
		\draw[thick,-to] (-0.1,0) -- (1.1,0) node[anchor=west]{$n$};
		\draw[thick,-to] (0,-0.1) -- (0,1.1) node[anchor=south]{proportion};
		\end{tikzpicture}
	\caption{The proportion of pure states $\ket{x} \in \C^n$ with $C^\prime_\tr(\ketbra{x}{x}) = 1$, for values of $n$ ranging from $2$ to $20$. Notice the log scale on the $y$-axis.}\label{fig:ctr_mod_dims}
\end{figure}
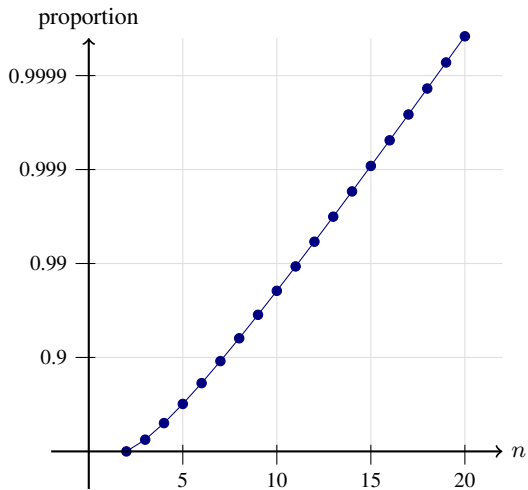

\begin{proof}[Proof of Theorem~\ref{thm:proportion_of_useless}]
    One way of generating Haar-uniform pure states $\ket{x} \in \C^n$ is to independently generate $2n$ $N(0,1)$-distributed (i.e., normal-distributed with mean~$0$ and variance~$1$) random variables $y_1,y_2,\ldots,y_n, z_1,z_2, \ldots, z_n$ and then set \cite{Mul59} \[\ket{x} = \frac{(y_1 + iz_1,y_2+iz_2, \ldots, y_n + iz_n)}{\|(y_1 + iz_1,y_2+iz_2, \ldots, y_n + iz_n)\|}.\]
    
    To start, we compute the probability that $|x_1| := \sqrt{|y_1|^2 + |z_1|^2} \geq 1/\sqrt{2}$, which is equivalent to \[|y_1|^2 + |z_1|^2 \geq (|y_2|^2 + \cdots + |y_n|^2) + (|z_2|^2 + \cdots + |z_n|^2).\] The sum of the squares of $k$ independent $N(0,1)$-distributed random variables follows a chi-squared distribution with $k$ degrees of freedom (facts like this one are contained in standard mathematical statistics textbooks like \cite{MathStatBook1}), so we are asking exactly for $P(X \geq Y)$, where $X \sim \chi^2_2$ and $Y \sim \chi^2_{2n-2}$. Then
    \begin{equation*}
        P(X \geq Y) = P\left(\frac{X}{Y} \geq 1\right) = P\left(\frac{X/2}{Y/(2n-2)} \geq n-1\right).
    \end{equation*}
    Well, $\frac{X/2}{Y/(2n-2)}$ follows an F-distribution with $2$ and $2n-2$ degrees of freedom, which is known to have probability density function
    \begin{equation*}
        f_{2,2n-2}(x) = \left(\frac{n-1}{x+n-1}\right)^n.
    \end{equation*}
    Thus we conclude that
    \begin{equation*}\begin{aligned}
        P(|x_1| \geq 1/\sqrt{2}) & = P\left(\frac{X/2}{Y/(2n-2)} \geq n-1\right) \\
        & = \int_{n-1}^\infty f_{2,2n-2}(x) \, dx \\
        & = -\left(\frac{n-1}{x+n-1}\right)^{n-1}\Bigg|_{n-1}^\infty \\
        & = 1/2^{n-1}.
    \end{aligned}\end{equation*}
    
    Since the $n$ events $\big\{|x_j| \geq 1/\sqrt{2}\big\}$ have the same probability regardless of $j$ and are mutually exclusive, we conclude that \[P(\max_j\{|x_j|\} \geq 1/\sqrt{2}) = nP(|x_1| \geq 1/\sqrt{2}) = n/2^{n-1},\]so the probability that $|x_j| \leq 1/\sqrt{2}$ for all $j$ (and thus $C^\prime_\tr(\ketbra{x}{x}) = 1$) is $1-n/2^{n-1}$.
\end{proof}

\section{IV. Higher-rank density matrices}\label{sec:higher_rank}

The modified trace distance of coherence behaves more like one might na\"{i}vely expect on general mixed states, but a similar ``almost-constant'' phenomenon seems to occur if we fix the rank of the density matrices (not necessarily equal to $1$) and let the dimension grow. For example, Figure~\ref{fig:ctr_mod_dims_rank} illustrates what proportion of density matrices $\rho$ (again, with respect to uniform Haar measure \cite{HaarExplain}) of a given rank in a given dimension have $C^\prime_\tr(\rho) = 1$.

\begin{figure}[!htb]
	\centering
	\def\x{\noexpand\x}    % Prevent \x from being expanded inside an \edef
	\edef\ctr{2*\x*pow(1 - pow(\x,2),0.5)}
	\begin{tikzpicture}[yscale=5,xscale=3.3333]
		% GRID
		\foreach \x in {1,...,6} \draw[color=gray!25] (0.25*\x,0) -- (0.25*\x,1.1);
		\foreach \y in {1,...,4} \draw[color=gray!25] (0,0.25*\y) -- (1.65,0.25*\y);
		
		% TICKS on AXES
		\draw (0.25,0.5pt) -- (0.25,-1pt) node[anchor=north] {\footnotesize 5};
		\draw (0.5,0.5pt) -- (0.5,-1pt) node[anchor=north] {\footnotesize 10};
		\draw (0.75,0.5pt) -- (0.75,-1pt) node[anchor=north] {\footnotesize 15};
		\draw (1,0.5pt) -- (1,-1pt) node[anchor=north] {\footnotesize 20};
		\draw (1.25,0.5pt) -- (1.25,-1pt) node[anchor=north] {\footnotesize 25};
		\draw (1.5,0.5pt) -- (1.5,-1pt) node[anchor=north] {\footnotesize 30};
		\draw (0.75pt,0.25) -- (-1.5pt,0.25) node[anchor=east] {\footnotesize 1/4};
		\draw (0.75pt,0.5) -- (-1.5pt,0.5) node[anchor=east] {\footnotesize 1/2};
		\draw (0.75pt,0.75) -- (-1.5pt,0.75) node[anchor=east] {\footnotesize 3/4};
		\draw (0.75pt,1) -- (-1.5pt,1) node[anchor=east] {\footnotesize 1};
		
		% CURVES
		\draw[color=myblue11] plot [smooth,tension=0.25] coordinates {(1.35,0.0) (1.4,0.0193) (1.45,0.0496) (1.475,0.0826) (1.5,0.1348)}; % rank 11
		\node at (1.568,0.1308) [rectangle,draw=white,label={center:\textcolor{myblue11}{$11$}}] {};

		\draw[color=myblue10] plot [smooth,tension=0.25] coordinates {(1.25,0.0) (1.3,0.0298) (1.35,0.0847) (1.4,0.1551) (1.45,0.2662) (1.5,0.4317)}; % rank 10
		\node at (1.568,0.4317) [rectangle,draw=white,label={center:\textcolor{myblue10}{$10$}}] {};

		\draw[color=myblue9] plot [smooth,tension=0.25] coordinates {(1.1,0.0) (1.15,0.0228) (1.2,0.0645) (1.25,0.1369) (1.3,0.2700) (1.35,0.4011) (1.4,0.5144) (1.45,0.6287) (1.5,0.7303)}; % rank 9
		\node at (1.568,0.7303) [rectangle,draw=white,label={center:\textcolor{myblue9}{$9$}}] {};
    
		\draw[color=myblue8] plot [smooth,tension=0.25] coordinates {(0.95,0.0) (1,0.0120) (1.05,0.0440) (1.1,0.1000) (1.15,0.1960) (1.2,0.3340) (1.25,0.4850) (1.3,0.6180) (1.35,0.7160) (1.4,0.7960) (1.45,0.8520) (1.5,0.9000)}; % rank 8
		\node at (1.5,0.875) [rectangle,draw=white,minimum size=1.4em,fill=white,label={center:\textcolor{myblue8}{$8$}}] {};
		
		\draw[color=myblue7] plot [smooth,tension=0.25] coordinates {(0.8,0.0) (0.85,0.0037) (0.9,0.0277) (0.95,0.0720) (1,0.1499) (1.05,0.2970) (1.1,0.4273) (1.15,0.5752) (1.2,0.6858) (1.25,0.7704) (1.3,0.8363) (1.35,0.8875) (1.4,0.9310) (1.45,0.9613) (1.5,0.9742)}; % rank 7
		\node at (1.34,0.875) [rectangle,draw=white,minimum size=1.4em,fill=white,label={center:\textcolor{myblue7}{$7$}}] {};

		\draw[color=myblue6] plot [smooth,tension=0.25] coordinates {(0.7,0.0) (0.75,0.0057) (0.8,0.0430) (0.85,0.1347) (0.9,0.2722) (0.95,0.4126) (1,0.5286) (1.05,0.6476) (1.1,0.7558) (1.15,0.8309) (1.2,0.8795) (1.25,0.9241) (1.3,0.9556) (1.35,0.9713) (1.4,0.9829) (1.45,0.9894) (1.5,0.9943) (1.55,0.9970) (1.6,0.9985) (1.65,0.9999)}; % rank 6
		\node at (1.17,0.875) [rectangle,draw=white,minimum size=1.4em,fill=white,label={center:\textcolor{myblue6}{$6$}}] {};
    
		\draw[color=myblue5] plot [smooth,tension=0.25] coordinates {(0.55,0.0) (0.6,0.0) (0.65,0.0237) (0.7,0.0881) (0.75,0.2271) (0.8,0.3729) (0.85,0.5190) (0.9,0.6393) (0.95,0.7458) (1,0.8333) (1.05,0.8820) (1.1,0.9227) (1.15,0.9526) (1.2,0.9704) (1.25,0.9811) (1.3,0.9890) (1.35,0.9937) (1.4,0.9976) (1.45,0.9986) (1.5,0.9994) (1.55,0.9999) (1.65,0.9999)}; % rank 5
		\node at (1.04,0.875) [rectangle,draw=white,minimum size=1.4em,fill=white,label={center:\textcolor{myblue5}{$5$}}] {};
    
		\draw[color=myblue4] plot [smooth,tension=0.25] coordinates {(0.45,0) (0.5,0.0073) (0.55,0.0610) (0.6,0.1659) (0.65,0.3302) (0.7,0.5073) (0.75,0.6434) (0.8,0.7512) (0.85,0.8385) (0.9,0.8901) (0.95,0.9360) (1,0.9601) (1.05,0.9756) (1.1,0.9860) (1.15,0.9927) (1.2,0.9950) (1.25,0.9976) (1.3,0.9990) (1.35,0.9999) (1.65,0.9999)}; % rank 4
		\node at (0.88,0.875) [rectangle,draw=white,minimum size=1.4em,fill=white,label={center:\textcolor{myblue4}{$4$}}] {};

		\draw[color=myblue3] plot [smooth,tension=0.25] coordinates {(0.35,0.0) (0.4,0.0694) (0.425,0.1294) (0.45,0.2014) (0.5,0.3483) (0.55,0.4992) (0.6,0.6517) (0.65,0.7599) (0.7,0.8333) (0.75,0.8820) (0.8,0.9227) (0.85,0.9526) (0.9,0.9704) (0.95,0.9811) (1,0.9890) (1.05,0.9937) (1.1,0.9976) (1.15,0.9986) (1.2,0.9994) (1.25,0.9999) (1.65,0.9999)}; % rank 3
		\node at (0.74,0.875) [rectangle,draw=white,minimum size=1.4em,fill=white,label={center:\textcolor{myblue3}{$3$}}] {};
    
		\draw[color=myblue2] plot [smooth,tension=0.25] coordinates {(0.225,0.0) (0.25,0.0323) (0.3,0.1923) (0.35,0.3686) (0.4,0.5376) (0.45,0.7101) (0.5,0.8160) (0.55,0.8821) (0.6,0.9308) (0.65,0.9619) (0.7,0.9776) (0.75,0.9895) (0.8,0.9931) (0.85,0.9965) (0.9,0.9978) (0.95,0.9992) (1,0.9999) (1.05,0.9999) (1.65,0.9999)}; % rank 2
		\node at (0.545,0.875) [rectangle,draw=white,minimum size=1.4em,fill=white,label={center:\textcolor{myblue2}{$2$}}] {};
    
		\draw[color=myblue] (0.1,0) -- (0.15,0.2506) -- (0.2,0.4998) -- (0.25,0.6871) -- (0.3,0.8125) -- (0.35,0.8911) -- (0.4,0.9377) -- (0.45,0.9648) -- (0.5,0.9806) -- (0.55,0.9892) -- (0.6,0.9942) -- (0.65,0.9969) -- (0.7,0.9982) -- (0.75,0.9991) -- (0.8,0.9995) -- (0.85,0.9997) -- (0.9,0.9999) -- (1.65,0.9999); % rank 1
		\node at (0.345,0.875) [rectangle,draw=white,minimum size=1.4em,fill=white,label={center:\textcolor{myblue}{$1$}}] {};
		\node at (0.18,0.875) [rectangle,draw=white,minimum size=1.4em,fill=white,label={center:\textcolor{myblue}{rank:}}] {};
		
		% AXES
		\draw[thick,-to] (-0.25,0) -- (1.65,0) node[anchor=west]{$n$};
		\draw[thick,-to] (0,-0.1) -- (0,1.1) node[anchor=south]{proportion};
		\end{tikzpicture}
	\caption{The proportion of mixed states $\rho \in M_n(\mathbb{C})$ with $C^\prime_\tr(\rho) = 1$, for values of $n$ ranging from $2$ to $30$ and various small ranks. Each curve corresponds to density matrices of a particular rank, which is indicated on the curve itself. The blue curve for rank-$1$ states is the same as in Figure~\ref{fig:ctr_mod_dims}, but not on a log scale.}\label{fig:ctr_mod_dims_rank}
\end{figure}
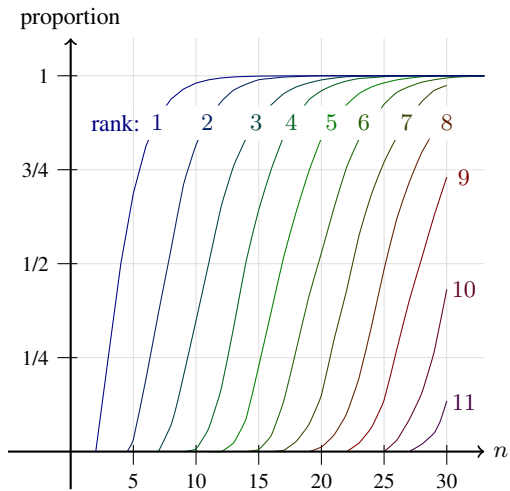

Numerically it seems that approximately 50\% of states $\rho \in M_n(\mathbb{C})$ with $\mathrm{rank}(\rho) \leq n/3$ have $C^\prime_{\tr}(\rho) = 1$. These numerics were computed using the YALMIP \cite{YALMIP} optimization package for MATLAB, and our code for computing the modified trace distance of coherence can be downloaded from \cite{MATLABcode}.

\section{V. Conclusions}\label{sec:conclusions}

We have shown that the modified trace distance of coherence is constant and equal to its maximum value of $1$ on all except for an exponentially-small ($n/2^{n-1}$ in dimension $n$) fraction of pure states, and we have provided numerical evidence that suggests a similar phenomenon occurs for density matrices of any fixed rank. It would be interesting to pin down this numerical observation rigorously, but we believe that the pure state result is enough to suggest that other measures of coherence, which attain their maximal value at essentially a unique state, should be preferred over the modified trace distance of coherence.\\

\noindent\textbf{Acknowledgements.} N.J.\ was supported by {NSERC} Discovery Grant number RGPIN-2016-04003. 
 C.-K.L.\ is an affiliate member of the Institute for Quantum Computing, 
University of Waterloo. He is an honorary professor of the University of Hong Kong and the Shanghai 
University. His research was supported
by USA NSF grant DMS 1331021, Simons Foundation Grant 351047, and NNSF of China
Grant 11571220. S.P.\ was supported by NSERC Discovery Grant number 1174582, the Canada Foundation for Innovation, and the Canada Research Chairs Program.


\begin{thebibliography}{50}
\bibitem{HaarExplain} More specifically, a rank-$r$ $\rho \in M_n(\mathbb{C})$ is chosen by generating a Haar-uniform $\ket{v} \in \C^r \otimes \C^n$ and tracing out the first subsystem.
\end{thebibliography}
\end{document}